\documentclass[final]{article}

\pdfoutput=1
\usepackage{fullpage}
\usepackage{lmodern}
\usepackage{authblk}

\usepackage[utf8]{inputenc}
\usepackage[colorlinks,citecolor=blue,linktocpage,breaklinks,hypertexnames=false,pdfpagelabels,draft=false]{hyperref}
\usepackage{microtype}

\usepackage[obeyDraft]{todonotes}
\usepackage[intlimits]{amsmath}
\usepackage{amssymb, amsfonts, amsthm}
\usepackage{enumerate}
\usepackage{quantum}

\theoremstyle{plain}
\newtheorem{lemma}{Lemma}
\newtheorem{thm}[lemma]{Theorem}
\newtheorem{prop}[lemma]{Proposition}
\newtheorem{cor}[lemma]{Corollary}
\newtheorem{definition}[lemma]{Definition}
\theoremstyle{remark}
\newtheorem{obs}[lemma]{Observation}

\newcommand{\norm}[1]{\left\lVert#1\right\rVert}
\newcommand{\du}[1]{#1^*}
\newcommand{\de}[1]{\ensuremath{\operatorname{d}\!{#1}}}
\DeclareMathOperator{\poly}{poly}
\DeclareMathOperator{\dist}{dist}
\DeclareMathOperator{\trace}{Tr}
\DeclareMathOperator{\diam}{diam}
\DeclareMathOperator{\supp}{supp}
\renewcommand{\epsilon}{\varepsilon}
\newcommand{\mcl}{\mathcal}
\newcommand{\mds}{\mathds}

\newcommand{\rhs}{r.h.s.\ }
\newcommand{\ie}{i.e.\ }

\usepackage[backend=bibtex,style=numeric-comp,doi=false,url=false,firstinits=true,maxbibnames=4]{biblatex}
\title{Area law for fixed points of rapidly mixing dissipative quantum systems}

\author[1,2]{Fernando G.\,S.\,L.\ Brand\~ao \thanks{\texttt{f.brandao@ucl.ac.uk}}}
\author[2,3]{Toby S.\ Cubitt \thanks{\texttt{tsc25@cam.ac.uk}}}
\author[4]{Angelo Lucia \thanks{\texttt{anlucia@ucm.es}}}
\author[7]{Spyridon Michalakis \thanks{\texttt{spiros@caltech.edu}}}
\author[4,5,6]{David Perez-Garcia \thanks{\texttt{dperezga@ucm.es}}}

\affil[1]{Quantum Architectures and Computation Group, Microsoft Research, \authorcr Redmond, WA, U.\,S.\,A. }
\affil[2]{Department of Computer Science, University College London, \authorcr Gower Street, London WC1E 6BT, United Kingdom}
\affil[3]{DAMTP, University of Cambridge, U.\,K.}
\affil[4]{Departamento de Análisis Matemático, Universidad Complutense de Madrid, Spain}
\affil[5]{IMI, Universidad Complutense de Madrid, Spain}
\affil[6]{ICMAT, C/ Nicolás Cabrera, Campus de Cantoblanco, 28049 Madrid}
\affil[7]{Institute for Quantum Information and Matter, Caltech, U.\,S.\,A. }

\date{\today}

\begin{document}
	\maketitle
	\begin{abstract}
		We prove an area law with a logarithmic correction for the mutual information for fixed points of local dissipative quantum system
		satisfying a rapid mixing condition, under either of the following assumptions:
		the fixed point is pure, or the system is frustration free.
	\end{abstract}

\section{Introduction}

One of the problems common to quantum information and condensed matter physics is understanding correlations and entanglement in many-body quantum states. This is motivated by the observation that many interesting states behave very differently from random states: while in the latter the entanglement entropy of a sub-region scales as the volume of the region, in the states that occur in quantum many-body systems this entanglement entropy is often seen to scale only with the boundary of the region. This surprising behaviour has been called the \emph{area law} (the terminology comes from 3D systems), and it is a well-studied conjecture that ground-states of (gapped) local Hamiltonians should satisfy an area law. Even in the case of critical systems and gapless Hamiltonians, evidence suggests that the groundstate still has a sub-volume growth of the entanglement, with a rate proportional to the boundary of the region times a logarithmic correction \cite{Masanes2009,1505.07106v1}\footnote{Some authors call such situation a \emph{violation} of the area law, while others speak of an area law with a logarithmic correction. We will use the second terminology, but we will use the term \emph{sub-volume} law for any faster growth rate of entanglement that is slower than the volume of the region.}.
While it has been formally proven only for 1D systems \cite{Hastings07,2013arXiv1301.1162A,PhysRevB.85.195145,Brandao2013, Brandao2015}, area laws have also been proven in specific cases in higher dimensions (harmonic lattice systems \cite{Plenio2005}, models satisfying Local Topological Quantum Order \cite{Spiros11}, perturbations of gapped Hamiltonians satisfying an area law \cite{1206.6900v2,1411.0680v1}, and with a logarithmic correction for fermionic systems \cite{Wolf2006,Gioev2006}), and are the subject of active research.

Recently, a different class of states arising in quantum many-body systems has been attracting attention in the quantum information literature: fixed points of (local) dissipative processes. More precisely, fixed points of semigroups of trace preserving, completely positive linear maps. The motivation is two-fold: on the one hand, such processes model most of the different types of noise that can be found in nature, and therefore provide a more realistic model for physical systems, since in practice no system will be completely isolated. On the other hand, proposals have been made to artificially engineer such dissipative interactions in order to have a determined quantum state as a fixed point,  effectively making them ``dissipative machines'' for producing useful/interesting quantum states \cite{Kraus08,verstraete09}. This dissipative state engineering has been experimentally shown to be a robust mechanism to maintain coherence \cite{PRL.107.080503,2010NatPh6.943B}.

A natural question then arises: is there an area law (either strict or with a logarithmic correction) in this context? Note that, since fixed points of dissipative evolutions are generically not pure, we must find another measure of entanglement or correlations, since local entropy is no longer a useful measure for mixed states (as the trivial example of the maximally mixed state shows). \cite{PhysRevLett.100.070502} proposed instead using the \emph{mutual information}, a measure of correlations between two parts of a quantum state. It has the advantage of coinciding with entanglement entropy for pure states, and it upper bounds operational measures of entanglement in the mixed-state case, such as the distillable entanglement \cite{bennett1996mixed}. In \cite{PhysRevLett.100.070502} it was also shown that thermal states of local Hamiltonians satisfy an area law for the mutual information. However, since thermal states do not cover all the possible fixed points of dissipative systems. Pinning down under which conditions an area law for the mutual information holds for general dissipative quantum many-body systems is thus an interesting open problem.

For Hamiltonian systems, the main assumption that is usually made is the presence of a spectral gap: a non-vanishing separation between the two lowest energy levels of the Hamiltonian. In the dissipative setting, instead of spectral assumptions, it is more natural to make assumptions on the speed of convergence of the dissipation towards its fixed point (a quantity that is not controlled by the spectrum alone \cite{spectralboundsquantum}), or equivalently on the so-called \emph{mixing time}. In this work we restrict to systems for which the mixing time scales logarithmically with the system size. In a previous paper \cite{ourselves} some of the authors showed that such systems -- which we called \emph{rapid mixing} -- are stable under local perturbations. Similarly to the gap of closed quantum systems, proving rapid mixing for dissipative systems is a daunting task. There are however some interesting key examples: state preparation of graph states \cite{Kastoryano12}, classical Glauber dynamics for the Ising model in 2D (in some range of parameters) \cite{martinelli-2d,Lubetzky-Sly}, among others.

Rapid mixing is implied in many cases by a well studied property of dissipative evolutions, namely the existence of a system-size independent Log-Sobolev constant for primitive reversible Liouvillians \cite{Gross-3,Gross-2,Gross-1,Quantum-Log-Sobolev,Yoshida2001223}. Under such assumptions, in \cite{Kastoryano-exp-decay} it was proven a bound on the mutual information of the form:
\begin{equation}
\label{eq:kastoryano-area-law}
 I(A:A^c) \le c \log \log \norm{\rho^{-1}} \abs{\partial A}.
\end{equation}
In the latter bound $\rho$ is the fixed point of the evolution, and therefore $\norm{\rho^{-1}}$ will usually depend at least exponentially on the total system size, and sometimes even worse. As recognized by the authors of \cite{Kastoryano-exp-decay}, this poses a serious problem in considering Eq.~\eqref{eq:kastoryano-area-law} a satisfactory area law. It indicates however that rapid mixing seems to be the required condition to have an area law in the dissipative setting. This is exactly what we prove in the present paper, with the following results:
\begin{enumerate}
	\item if the system satisfies rapid mixing and the fixed point is pure, then it satisfies an area law with a logarithmic correciton for the entanglement entropy;
	\item if the system satisfies rapid mixing and is frustration free, meaning that the local terms of the Liouvillian share a common steady state, then such fixed point satisfies an area law with a logarithmic correciton for the mutual information.
\end{enumerate}

Compared with \eqref{eq:kastoryano-area-law}, the bounds we obtain do not have any dependence on the total system size. Moreover we do not require primitivity or reversibility of the generators of the evolution, and we only require rapid mixing instead of a system-size independent Log-Sobolev constant (a strictly weaker assumption, since the Log-Sobolev constant is undefined for non-primitive Liouvillians).

It is known that there is a connection between area laws and decay of correlations \cite{Brandao2013,PhysRevLett.100.070502}. Therefore it does not come as a surprise that with the tools we have developed for proving the area law we can also prove a decay of correlations measured with the mutual information. It is worth noting that, with the results available in the literature and to the best of our knowledge, it is not possible to derive the area law from the type of decay of correlations we will show here.

As we have mentioned earlier, for groundstates of closed systems a logarithmic correction is usually considered a signature of a gapless Hamiltonian. For open systems and mixed states the situation is less clear: already in \cite{PhysRevLett.100.070502} it was shown that thermal states of local Hamiltonian satisfy an area law without a logarithmic correction irrespective of the gap of the Hamiltonian. For these states the bound we obtained is therefore not optimal. We do not know whether there exists systems which saturate our bound, or if instead the correction is only an artefact of the proof. If there exist systems that saturate our bound, it would then imply that their fixed point have a very interesting property: while still satisfying an exponential decay of correlations, they do not satisfy an area law without a logarithmic correction. Having an example of such state which can also be efficiently prepared with a dissipative process would be interesting on its own, as it could lead to new insight on the relationship between area laws and decay of correlations.

We conjecture that rapid mixing alone, without any additional assumptions, should imply an area law for mutual information, but we do not have a formal proof. The fact that we have two different proofs of an area law, requiring different extra assumptions (pure fixed point on the one hand, frustration freeness on the other) is strong evidence for this conjecture.

The paper is organized as follows. In section \ref{section:notation} we set up the problem and introduce the necessary notation and definitions. In section \ref{section:localization} we prove two lemmas regarding localization properties of the fixed point of the dissipative maps: Lemma~\ref{lemma:rapidmixing-localization} is based only on the rapid mixing assumption, and will be used in section~\ref{section:correlation-decay} to prove decay of correlations and in section~\ref{section:correlation-arealaw} to prove the area law for pure fixed points; Lemma \ref{lemma:localization} instead requires the extra assumption of frustration freeness, and will be used in section~\ref{section:correlation-arealaw} to prove the area law for the mutual information in the case of mixed fixed points.


\section{Setup and notation}
\label{section:notation}
Let $\mcl H_{AB} = \mcl H_{A} \otimes \mcl H_{B}$ be a finite dimensional complex Hilbert space,
representing a bipartite quantum system. We denote by $\mcl B(\mcl H_{AB})$ the space of bounded linear operators on $\mcl H_{AB}$.
 A \emph{state} is given by a positive semi-definite operator $\rho_{AB} \in \mcl B(\mcl H_{AB})$,
normalized to have trace equal to one. The \emph{reduced density matrix} of the subsystem A (resp. B) will be
denoted by $\rho_A$ (resp. $\rho_B$), and it is given by $\rho_A = \tr_A \rho_{AB}$ (resp. $\rho_B = \tr_B \rho_{AB}$),
where the partial trace $\tr_A$ is defined to be the unique linear operator $\tr_A : \mcl B(\mcl H_{AB}) \to \mcl B(\mcl H_{B})$
such that $\tr_A(x\otimes y) = y \tr(x)$ for all $x$ in $\mcl B(\mcl H_{A})$ and all $y$ in $\mcl B(\mcl H_{B})$ 
($\tr_B$ is similarly defined).

We will use the standard Dirac notation for Hilbert spaces, denoting vectors as 
$\ket{\phi}$, adjoint vectors as  $\bra{\phi}$, $\braket{\phi}{\psi}$ for the scalar product,
and $\ketbra{\phi}{\psi}$ for rank one linear maps. The canonical basis will be indexed by natural numbers starting from zero:
$\ket{0}, \ket{1}, \dots \ket{n}$.

We will denote by $\poly(x)$ any polynomial in the variable $x$ with real coefficients and arbitrary degree.

\subsection{Measures of correlations}
Given a state $\rho_{AB} \in \mcl B(\mcl H_{AB})$ of a bipartite system, there are a number of possible measures of how ``distant'' the state $\rho_{AB}$
is from being a product state, \ie of the form $x \otimes y$ for some $x$ in $\mcl B(\mcl H_A)$ and $y$ in $\mcl B(\mcl H_B)$.
Since a product state represents a system in which measurements over the subsystem $A$ are independent of measurements over the subsystem $B$, we will talk of \emph{correlation measures} between subsystem $A$ and $B$. We will need to define and use three of such measures.
We will follow the same terminology of \cite{Kastoryano-exp-decay}.

\begin{definition}[Correlation measures]\hfill\\
\label{def:correlations}
  \begin{itemize}
    \item \emph{Covariance correlation}:
      \begin{multline*}
      C(A : B)  
      = \max_{\substack{M \in \mcl B(\mcl H_A),N \in \mcl B(\mcl H_B)\\ \norm{M} \le 1, \norm{N} \le 1}}
      \abs{ \expect{M\otimes N} - \expect{M}\expect{N} } \\
        = \max_{\substack{M \in \mcl B(\mcl H_A),N \in \mcl B(\mcl H_B)\\ \norm{M} \le 1, \norm{N} \le 1}}
      \abs{ \trace \left[ M \otimes N (\rho_{AB} - \rho_A \otimes \rho_B)\right] } ;
    \end{multline*}
    where $\expect{O} = \trace (O \rho_{AB})$ is the expectation value of the observable $O$ acting on $\rho_{AB}$.

    \item \emph{Trace distance correlation}:
      \begin{align*}
        T(A : B) &= \max_{\substack{F \in \mcl B(\mcl H_{AB})\\ \norm{F} \le 1}}
        \abs{ \trace \left[F (\rho_{AB} - \rho_A \otimes \rho_B)\right] } \\
        &= \norm{\rho_{AB} - \rho_A \otimes \rho_B}_1 .
      \end{align*}
    \item \emph{Mutual information correlation}:
      \[ I(A : B) = S(\rho_A) + S(\rho_B) - S(\rho_{AB}) ;\]
      where $S(\rho) = - \trace (\rho \log_2 \rho)$ is the von Neumann entropy of the state $\rho$.
    \end{itemize}
\end{definition}
When it is not be clear from context which state $\rho_{AB}$ we are considering, we indicate it in a subscript to avoid ambiguity, and write $C(A:B)_{\rho}$, $T(A:B)_{\rho}$ and $I(A:B)_{\rho}$.

As should be clear from the definition, $C(A:B)$ is always upper bounded by $T(A:B)$. Moreover, by Pinsker’s inequality \cite{nielsen-chuang}:
\[ C(A : B) \le T(A : B) \le 2 \sqrt{ I(A : B) } .\]

Therefore, mutual information is the strongest correlation measure. It is also a well known consequence of the Alicki-Fannes-Audenaert inequalities \cite{2007JPhA...40.8127A, MR0345574,MR2043448} that there is a non-linear inverse relationship between trace distance and mutual information: there is a differentiable function $f(x)$ vanishing at zero such that $I(A:B) \le f(T(A:B))$. Such non-linear equivalence between the two measures will allow us to take bounds on $T(A:B)$ (which in the context of our assumptions will be easier to deduce) and obtain information on the behavior of $I(A:B)$.

We state this result in a form that will be more convenient for us and, for the sake of completeness, present a short proof of it.

\label{sec:appendix-mutualinfo}
For $x \in [0,1]$, $h_b(x) = -x \log_2 x - (1-x) \log_2(1-x)$ denotes the binary entropy function.

\begin{thm}
The following inequalities hold:
\begin{description}
\item
  [Fannes-Audenaert \cite{2007JPhA...40.8127A, MR0345574} ]\hfill\\
  Let $\rho, \sigma \in \mcl B(\CC^d)$, and let $\delta = \norm{\rho-\sigma}_1 < 1$. Then
  \begin{equation}
    \label{eq:fannes-audenaert}
    \Abs{ S(\rho) - S(\sigma) } \le 2 \delta \log_2(d-1) + 2 h_b(\delta),
  \end{equation}
\item
  [Alicki-Fannes \cite{MR2043448}]\hfill\\
  Let $\rho^{AB}, \sigma^{AB} \in \mcl B(\CC^{d_A}\otimes \CC^{d_B})$, and let $\delta = \norm{\rho^{AB} - \sigma^{AB}}_1 < 1$. Then
  \begin{equation}
    \label{eq:alicki-fannes}
    \Abs{ S( \rho^{AB} | \rho^B) - S( \sigma^{AB} | \sigma^B ) } \le 4 \delta \log_2 d_A + 2 h_b(\delta),
  \end{equation}
  where the conditional Von Neumann entropy $ S( \rho^{AB} | \rho^B)$ is defined as
  $ S( \rho^{AB} | \rho^B) = S(\rho^{AB}) - S(\rho^B)$.
\end{description}
\end{thm}

Combining the two previous inequalities, we will obtain the desired non-linear bound on $I(A:B)$.

\begin{cor}
  Let $\delta = \norm{\rho^{AB} - \sigma^{AB}}_1 < 1$. Then
  \begin{equation}
    \label{eq:alicki-mutual}
    \Abs{ I(A:B)_\rho - I(A:B)_\sigma } \le 6 \delta \log_2 d_A + 4 h_b(\delta) .
  \end{equation}
  In particular, if we take $\sigma^{AB} = \rho^A \otimes \rho^B$, we have that
  $ I(A:B)_\sigma = 0$, \mbox{$\delta = T(A:B)_\rho$}, and thus
  \begin{equation}
   \label{eq:fannes-mutual}
    I(A:B)_\rho \le 6 T(A:B)_\rho \log_2 d_A + 4 h_b( T(A:B)_\rho).
  \end{equation}
  \end{cor}

\begin{proof}
  Note that $\norm{\rho^{AB} - \sigma^{AB}}_1 \le \delta \Rightarrow \norm{\rho^A - \sigma^A}_1 \le \delta$. Applying triangle inequality and equations \eqref{eq:fannes-audenaert} and \eqref{eq:alicki-fannes} gives
  \[ \Abs{ I(A:B)_\rho - I(A:B)_\sigma } \le
    \Abs{ S(\rho^A) - S(\sigma^A)} + \Abs{ S(\rho^{AB} | \rho^B) - S(\sigma^{AB} | \sigma^B )} .\]
\end{proof}


\subsection{Many-body quantum systems}

Let us now recall the standard definitions and the common notation for many-body quantum systems.

We will consider a quantum system defined on the square lattice $\Gamma = \ZZ^D$ equipped with the graph metric,
where at each site $x \in \Gamma$ we associate a finite-dimensional complex Hilbert space $\mcl H_x$.
We choose to work with a square lattice for simplicity of exposition, but the results presented can be generalized
straightforwardly to graphs with polynomial growth (\ie with balls size growing polynomially with the diameter).
The ball centered at $x$ of radius $r$ will be denoted by $b_x(r)$.
We will use the following convention: given a subset $A \subset \Lambda$, we will denote by
$A(s)$ the smallest disjoint union of balls containing $\{ x \in \Lambda | \dist(x,A) \le s \}$.
For each finite subset $\Lambda \subset \Gamma$ of the lattice, we will associate a Hilbert space
$\mcl H_\Lambda = \otimes_{x\in \Lambda} \mcl H_x$ and an algebra of observables $\mcl A_\Lambda = \mcl B(\mcl H_\Lambda)$.
We will equip $\mcl A_\Lambda$ with the Hilbert-Schmidt scalar product $\langle A, B\rangle = \trace (A^* B)$.

A linear map $\mathcal T: \mathcal A_\Lambda \to \mathcal A_\Lambda$ will be called a \emph{superoperator} \cite{Wolf11}
to stress the fact that it is an operator acting on operators. Its \emph{support} is defined to be
the minimal set $\Lambda^\prime \subseteq \Lambda$ such that $\mathcal T = \mathcal T^\prime \otimes \identity$, where $\mathcal T^\prime \in \mathcal B(\mathcal A_{\Lambda^\prime})$. Positivity is defined as usual for linear maps: $\mcl T$ is said to be positive if it maps positive operators to positive operators. $\mathcal T$ is called \emph{completely positive} if $\mathcal T \otimes \identity : \mathcal A_\Lambda \otimes M_n \to \mathcal A_\Lambda \otimes M_n$ is positive for all $n \ge 1$. Finally, we say that $\mathcal T$ is trace preserving if $\trace \mathcal T(\rho) =\trace \rho$ for all $\rho \in A_\Lambda$.

A \emph{dissipative evolution} for a quantum system is given by a one-parameter continuous semigroup of
completely positive and trace preserving superoperators $\{ T_t : \mcl A_\Lambda \to \mcl A_\Lambda \}_t$. If $\rho \in \mcl B(\mcl H_\Lambda)$ is the state of the system at time zero, then the evolution of $\rho$ at time $t \ge 0$ is given by $\rho(t) = T_t(\rho)$. The assumptions on $T_t$ guarantee that $\rho(t)$ is again a state, \ie a positive and trace one operator. This is usually called the Schrödinger picture.

We will make use of the following norm for superoperators:
\begin{equation}
   \norm{T}_{\diamond} = \sup_n \norm{T\otimes \identity_n}_{1\to 1} = \sup_n \sup_{ \substack{X \in \mathcal A_\Lambda \otimes M_n \\ X \neq 0 }} \frac{ \norm{T\otimes \identity_n(X)}_1}{\norm{X}_1} .
\end{equation}
Dissipative maps are contractive with respect to such norm, in the sense that $\norm{\mcl T}_\diamond \le 1$.

Given a semigroup of dissipative maps $\{T_t\}$, it has a generator $\mcl L:\mcl A_\Lambda \to \mcl A_\Lambda$ which satisfies
$ \frac{\de{}}{\de t} T_t(\rho) = \mcl L T_t(\rho)$.
Such superoperator is called a \emph{Lindbladian} or \emph{Liouvillian} (we will use the former). The assumptions made on $T_t$ force a particular structure on $\mcl L$, which is called the \emph{Lindblad form} \cite{Kossakowski,Lindblad}. A superoperator $\mcl L$ is said
to be in the Lindblad form if it can be written as
\begin{equation}
\mcl L(\rho) = i [\rho, H] + \sum_j L_j \rho \du L_j - \frac{1}{2} \{ \du L_j L_j, \rho \} ,
\end{equation}
where $H$ is an Hermitian operator, $(L_j)_j$ are arbitrary operators (called the Lindblad operators), $[\cdot,\cdot]$ denotes the commutator and $\{\cdot,\cdot\}$ the anti-commutator. We refer to \cite{Wolf11,breuer2002theory} for details on the theory of Lindblad operators.

As shown in \cite{Wolf11}, given a semigroup of dissipative maps $\{ T_t \}_t$ we can define a new map $T_\infty$ representing the ``infinite time limit'' of the evolution, or in other words the projector onto the space of fixed points of the evolution. $T_\infty$ is
again a completely positive, trace preserving superoperator, since it can be obtained as $\lim_{N \to \infty} \frac 1 N \sum_{n=0}^N T_{n}$.

\subsection{Uniform families}
Given a generator $\mcl L$, we can decompose it as a sum of \emph{local terms}, i.e. terms which are still of the Lindblad form but with controlled support:
\[ \mathcal L = \sum_{Z \subset \Lambda} \mathcal L_Z , \quad \supp \mathcal L_Z = Z, \quad \mcl L_Z \text{ is Lindbladian}.\]
When $\norm{\mcl L_Z}_\diamond$ is decaying with $\diam Z$, we will generically say that the evolution is \emph{local}. More stringent assumptions on the decay rate of the norms of the local generators will be required, and are formalized in Assumptions \eqref{eq:assumption-a1} and \eqref{eq:assumption-a2} in this section.

Since we are interested in dissipative evolutions defined on increasing sequences of lattices, and how their properties depend on the lattice size (often referred to as the \emph{system size}), we need to define a meaningful way of growing the evolutions with the lattice size, adding and modifying the necessary generator terms appropriately. As presented in \cite{ourselves}, the following definition of \emph{uniform families} of dissipative evolutions is one solution to this, which is general enough to cover a wide range of models and situations.

\begin{definition}\label{defn:boundary-condition}
  Given $\Lambda \subset \Gamma$, a \emph{boundary condition} for $\Lambda$ is given by a Lindbladian $\mathcal B^{\partial \Lambda} = \sum_{d\ge 1} \mcl B^{\partial \Lambda}_d$, where $\supp B^{\partial \Lambda}_d \subset \partial_d \Lambda :=  \{ x \in \Lambda \,|\, \dist(x,\Lambda^c) \le d \}$.
\end{definition}

\begin{definition}
  \label{def:uniform-family}
  A \emph{uniform family} of Lindbladians is given by the following:
  \begin{enumerate}[(i)]
  \item \textit{infinite Lindbladian}: a local Lindbladian $\mathcal M$ defined all of $\ZZ^D$: $\mcl M = \sum_{Z\subset \ZZ^D} \mcl M_Z$;
  \item \textit{boundary conditions}: a family of {\it boundary conditions} $\{ \mcl B^{\partial \Lambda} \}_{\Lambda}$, where $\Lambda = b_u(L)$, for each $u\in \ZZ^D$ and $L \ge 0$.
  \end{enumerate}
\end{definition}

\begin{definition}
  A local Lindbladian $\mcl L = \sum_{Z\subset \ZZ^D} \mcl L_Z$ is said to be \emph{translationally invariant} if $\mcl L_{Z+u} = \mcl L_{Z},\, \forall u \in \ZZ^D$.

  We say that a uniform family $\mcl L = \{ \mcl M, \mcl B\}$ is \emph{translationally invariant} if $\mathcal M$ is translationally invariant, and moreover $\mathcal B^{\partial b_u(L)}$ is independent of $u$.
\end{definition}

Given a uniform family $\mcl L = \{ \mcl M, \mcl B\}$, we fix the following notation for evolutions defined on $\Lambda = b_u(L) \subset \Gamma$:
\begin{align}
\mathcal L^\Lambda =  \sum_{Z \subset \Lambda} \mathcal M_Z &\quad \text{``open boundary'' evolution} ; \\
\mathcal L^{\overline \Lambda} = \mathcal L^\Lambda + \mathcal L^{\partial \Lambda} &\quad \text{``closed boundary'' evolution},
\end{align}
with the respective evolutions $T_t^\Lambda = \exp(t \mathcal L^\Lambda)$ and  $T_t^{\overline \Lambda} = \exp(t \mathcal L^{\overline \Lambda})$.

Until now we have made no specific assumption on the decay rate of the norms of the local generators. As mentioned above, in order to meaningfully talk about locality of the evolution, we need to impose that $\norm{\mcl L_Z}_\diamond$ is decaying with $\diam Z$. The rate at which such function decays classifies the system into one of the specific cases more usually considered in the literature: compactly supported (usually called finite-range interactions), exponentially decaying, super-polynomially decaying, power-law decaying, etc. We will take a more general approach, and will simply assume from now on that our family of Liouvillians satisfies the following assumptions.

\begin{definition}[Lieb-Robinson Assumptions]
\label{def:lr-assumptions}
There exists an \emph{increasing} function $\nu(r)$ satisfying $\nu(x+y)\le \nu(x) \nu(y)$,
such that the following conditions hold:
\begin{equation}
  \label{eq:assumption-a1}
  \tag{A-1}
  \sup_{x\in \Gamma} \sum_{Z \ni x} \norm{\mcl M_Z}_\diamond \abs{Z} \nu(\diam Z) \le v < \infty,
\end{equation}
\begin{equation}
  \label{eq:assumption-a2}
  \tag{A-2}
  \sup_{x\in \Gamma} \sup_{r} \nu(r) \sum_{d = r}^N \norm{\mcl B_d^{\partial B(x,N)}}_\diamond  \le \poly(N).
\end{equation}
\end{definition}

If $\norm{\mcl M_Z}_\diamond$ is exponentially decaying or is compactly supported, then one can take $\nu(r) = \exp(\mu  r)$ for some positive $\mu$. On the other hand, if $\norm{\mcl M_Z}_\diamond$ decays only polynomially, then we must take $\nu(r) = (1+r)^{\mu}$. In the latter case, the Lieb-Robinson bounds. only hold if $\mu$ is bigger than a constant depending on the geometrical dimension of the lattice $\Gamma$, \ie $D$ if $\Gamma = \ZZ^D$. The details of when we can apply Lieb-Robinson bounds in this case can be found in our previous work \cite{ourselves}.
From now on we will simply assume that $\nu(\cdot)$ decays sufficiently fast for the Lieb-Robinson bounds to apply.

\subsection{Frustration freeness}
The following definition is inspired by the analogous concept defined for closed systems and Hamiltonian dynamics. It captures the idea that a fixed point of a local evolution might or might not be \emph{locally steady}. The local dissipative terms could in general have a non-trivial action on the fixed point, so that it is only the sum of such local effects that adds up to zero and leaves the state invariant. Assuming that the system is frustration free means excluding such cases.
\begin{definition}
We say that a uniform family $\mcl L = \{ \mcl M, \mcl B\}$ satisfies \emph{frustration freeness} (or is \emph{frustration free}) if for all $\Lambda$ and all
fixed points $\rho_\infty$ of $T_t^{\bar \Lambda}$
\begin{equation}
  \mcl M_Z (\rho_\infty) = 0 \quad \forall Z \subset \Lambda.
\end{equation}
\end{definition}
\begin{obs}
  If $\mcl L = \{ \mcl M, \mcl B\}$ is frustration free, then each fixed point of $T_t^{\bar \Lambda}$ is also a fixed point of $T_t^{\Lambda}$.
\end{obs}

While not true in general, the frustration freeness condition is satisfied by a large class of interesting dissipative systems, as the following examples  show:
\begin{enumerate}
\item Dissipative state engineering procedures defined in \cite{verstraete09,Kraus08};
\item Locally reversible classical Markov chains;
\item Locally detailed balanced quantum Markov processes, and in particular Gibbs samplers for commuting Hamiltonians \cite{arxiv1409.3435}.

\end{enumerate}


\section{Localization results}

A well known property of many-body systems with local interactions, either dissipative or Hamiltonian, is the existence of a finite speed of propagation. This describes how the support of a localized observable spreads in time during the evolution: up to an exponentially small correction, the support spreads linearly with time. The finite velocity at which such linear growth occurs is often called a Lieb-Robinson velocity, or sometimes a group velocity. (It is a property of the model and not a consequence of some relativistic effect -- we are considering only non-relativistic models here.) While the original work focused on Hamiltonian systems and groups of automorphisms \cite{robinson1968,lieb1972}, the existence of such finite speed of propagation in the lattice has been generalized to dissipative evolutions \cite{Poulin10,Nachtergaele12}, and in \cite{ourselves} we showed that the definition of boundary condition we have given allows us to recover the same type of localization properties of the evolution.

Nonetheless, all the Lieb-Robinson localization bounds have a time-dependency, becoming worse as time increases, until the bound they provide becomes trivial and does not give any information at all about the properties of the fixed point. In the next section we want to produce results which might be interpreted as ``infinite time'' versions of Lieb-Robinson bounds. To do so, we will need to make an extra assumption on the evolution: we will assume that the convergence to the fixed point is fast, in the sense that scales logarithmically with the system size. This is formalized in the definition of \emph{rapid mixing} below. Such a definition can be in some cases relaxed to allow convergence which is only scaling sub-linearly with respect to system size. (We will not pursue such generalizations here, and instead refer to \cite{ourselves} for guidance on which changes are necessary to the results below.)

The localization lemma that we prove, Lemma~\ref{lemma:rapidmixing-localization}, will be only be sufficient to prove an area law for the pure fixed point case, not for the mutual information. Therefore, we also prove a stronger result, Lemma~\ref{lemma:localization}, for which we will need to add the extra hypothesis of frustration freeness.

\label{section:localization}
\subsection{Rapid mixing}
In this section, we want to briefly recall a result proven in \cite{ourselves}. We start by recalling the definition
of \emph{rapid mixing}.

\begin{definition}[Rapid mixing]
	\label{def:rapid-mixing}
	Let $\{T_t^\Lambda\}_\Lambda$ be a family of dissipative maps, we say it satisfies rapid mixing if
	there exist $c, \gamma, \delta >0$ such that
	\begin{equation}
		\label{eq:rapid-mixing}
		\sup_{\substack{\rho \ge 0 \\ \trace \rho\, =1 }} \norm{T^\Lambda_t(\rho)- T^\Lambda_{\infty} (\rho)}_1 \le c \abs{\Lambda}^\delta \ e^{-t \gamma} .
	\end{equation}
\end{definition}

In analogy to the spectral gap for Hamiltonians, proving that a family of Lindbladians is rapid mixing is not an easy task.
Nonetheless, there exists a large class of interesting systems for which we already have mixing time estimates that imply rapid mixing:
\begin{enumerate}
	\item (Trivially) non-interacting particle systems.
	\item Dissipative state engineering for graph states \cite{Kastoryano12}.
	\item Quantum and classical Markov processes satisfying a Log-Sobolev inequality \cite{Quantum-Log-Sobolev}.
			This includes in particular Glauber dynamics for the Ising model in 2D, either above the critical temperature or with non-zero magnetic field \cite{martinelli-2d}.
\end{enumerate}

\begin{lemma}
 	\label{lemma:rapidmixing-localization}
	Let $\mcl L = \{ \mcl M, \mcl B\}$ be a uniform family of dissipative evolutions that
	satisfies rapid mixing, and suppose each $T_t^{\bar \Lambda}$ has a unique fixed point and no other periodic points.
	Fix a $\Lambda$ and let $\rho_\infty$ be the unique fixed point of $T^{\bar \Lambda}_t$.
	Given $A \subset \Lambda$, for each $s \ge 0$ denote by $\rho^s_\infty$ the unique fixed point of $T_t^{\bar A(s)}$.

	 Then we have:
	\begin{equation}
	  \norm{ \trace_{A^c} ( \rho_\infty - \rho_\infty^s ) }_1 \le \abs{A}^\delta \Delta_0(s),
	\end{equation}
	for some fast-decaying function $\Delta_0(s)$ and some positive constant $\delta$.
\end{lemma}

The decay rate of $\Delta_0(s)$ is in the same class as $\nu^{-1}(s)$,
where $\nu(s)$ is defined by Assumptions~\eqref{eq:assumption-a1} and~\eqref{eq:assumption-a2}:
it is exponential if $\nu^{-1}(s)$ is exponential, polynomial if $\nu^{-1}(s)$ is polynomial. In the latter case the degree of the polynomial
controlling the decay is smaller than that of $\nu^{-1}(s)$, but the loss is independent of the system size -- again, this corresponds to requiring a sufficiently fast-decaying $\nu^{-1}(s)$.

\subsection{Localizing with frustration freeness}
In the following section, we want to show a property of the fixed points of a uniform family of Lindbladians
verifying frustration freeness. We want to study the behavior of a system when it is prepared and started in a state,
which is the fixed point of the same family but of a slightly smaller region.
A reasonable guess would be that frustration freeness implies that the evolution should be localized ``around the boundary'', 
and that for short times nothing at all would happen in the ``bulk'' (where the state is left invariant
by the local interaction terms, because of frustration freeness). This intuition is formalized in the following lemma

\begin{figure}
	\caption{The construction of the sets $B$ and $R$.}
	\label{fig:localization}
	\centering
	\includegraphics{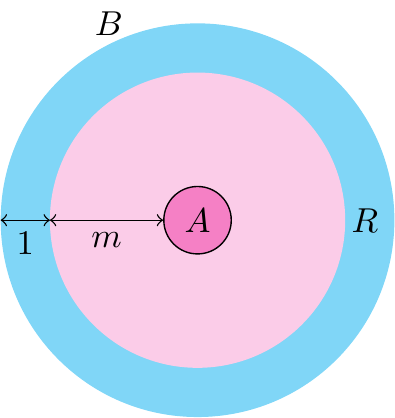}
\end{figure}

\begin{lemma}
	\label{lemma:localization}
	Let $ \mcl L = (\mcl M, \mcl B)$ an uniform family of Lindbladians, satisfying frustration freeness.
	Let $A \subset \Gamma$ be a finite region, and fix a positive natural number $m$. Let $B = A(m+1)$,
	$R = A(m+1)\setminus A(m)$ and
	$\rho_\infty^m$ a fixed point of $T_t^{\bar A(m)}$ and $\tau$ an arbitrary state on $R$ (see figure~\ref{fig:localization}).
	\begin{equation}
		\label{eq:localizing}
		  \norm{\left( T_t^{\bar B} - T_t^{ \bar B\setminus A} \right)(\rho_\infty^m \otimes \tau)}_1 \le
		  \poly(m) \nu^{-1}(m) \left[ e^{vt} -1 + t \right] ;
	\end{equation}
	where $T_t^{\bar B \setminus A}$ denotes the evolution
	generated by
	\[ \mcl L^{\bar B \setminus A} = \sum_{Z \subset {B\setminus A}} \mcl M_Z + \sum_{d \le m+1} \mcl B^{\partial B}_d .\]
\end{lemma}

In order to prove such result, we will first prove a Lieb-Robinson-type of lemma. Denote by
$ \rho(t) = T_t^{\bar B}(\rho_\infty^m \otimes \tau)$.
For each $X \subset B$, denote $\mds L_X$ the algebra generated by $\{ \mcl M_Z \,|\, Z \subset X \}$, which is the set of interactions terms of $\mcl L^{B}$ whose support is contained in $X$.

\begin{lemma}
	Under the same assumptions of Lemma \ref{lemma:localization},
	for each $K \in \mathds L_X$, the following ``Lieb-Robinson-like'' bound holds
	for some positive $v$:
	\begin{equation}
		\label{eq:lieb-rob}
		 \norm{K(\rho(t))}_1 \le \poly(m) \abs{X}\norm{K}_\diamond (e^{vt}+t-1) \nu^{-1}(\dist(X,R)) .
	\end{equation}
\end{lemma}
\begin{proof}
Denote
$ C(Z,t) = \sup_{T \in \mathds L_Z} \frac{\norm{T(\rho(t))}_1}{\norm{T}_\diamond}$.
Frustration-freeness implies that $C(Z,0)$ is $0$ if $Z\cap R = \emptyset$ (since $\trace_R \rho(0) = \rho_\infty^m$),
 while is bounded by $1$ otherwise. Moreover, let $ \Delta(r) =  \sum_{d \ge r} \norm{\mcl B_d^{\partial B}}_\diamond $
 and for each $Z \subset B$, let $\delta(Z) = \Delta(\dist(Z,R))$.
 Assumption~\eqref{eq:assumption-a2} implies that $\sup_r \nu(r) \Delta(r) \le \poly(m)$.

We are now going to replicate the proof technique of Lieb-Robinson bounds: denote $ \tilde{\mcl L}^X = \mcl L^{\bar B} - \mcl L^{\bar B \setminus X}$, and notice that, since they have disjoint support, $[K, \mathcal L^{\bar B \setminus X} ] =0$. Then
\[
\frac{\de{}}{\de t} K(\rho(t)) =
K \mathcal L^{\bar B} \rho(t) =
\mathcal L^{\bar B \setminus X} K(\rho(t)) + K \tilde{\mathcal L}^X(\rho(t)) ,
\]
and consequently,
\[ K(\rho(t)) =  e^{t \mathcal L^{\bar B \setminus X}} K(\rho(0)) + \int_0^t  e^{(t-s) \mathcal L^{\bar B \setminus X}} K \tilde{\mathcal L}^X \rho(s) \de s .  \]
By taking norms
\[ \norm{K(\rho(t))}_1 \le \norm{K(\rho(0))}_1 + \norm{K}_\diamond \int_0^t \norm{\tilde{\mathcal L}^X \rho(s)}_1 \de s ; \]
and thus
\begin{align}
C(X,t) &\le
	C(X,0) + \sum_{Z\cap X \neq \emptyset} \int_0^t \norm{\mathcal M_Z}_\diamond C(Z,s) \de s
	+ \sum_{d > m+1} \int_0^t \norm{\mathcal B^{\partial B}_d}_\diamond \de s  \nonumber \\
&\le
	C(X,0) + \delta(X) t + \sum_{Z \cap X \neq \emptyset} \norm{\mathcal M_Z}_\diamond \int_0^t C(Z,s) \de s.
	\label{eq:lr-recursion-step}
\end{align}
By recursively applying equation~\eqref{eq:lr-recursion-step}, we obtain that
\[ C(X,t) \le \sum_{n=0}^\infty [a_n + b_n] \frac{t^n}{n!} \]
where $a_0 = C(X,0)$,
\[ a_n = \sum_{Z_1 \cap X \neq 0} \dots \sum_{Z_n \cap Z_{n-1} \neq 0} \norm{\mcl M_{Z_1}} \dots \norm{\mcl M_{Z_n}} C(Z_n,0) ,\]
$b_0 = 0$, $b_1 = \delta(X)$ and
\[ b_{n+1} = \sum_{Z_1 \cap X \neq 0} \dots \sum_{Z_{n} \cap Z_{n-1} \neq 0} \norm{\mcl M_{Z_1}} \dots \norm{\mcl M_{Z_n}} \delta(Z_{n}). \]
Let us bound the two coefficients independently. The coefficients $a_n$ are threated in the same way as is done in the standard proof of Lieb-Robinson bounds \cite{Hastings10}: recalling that $C(Z,0)$ is zero unless $Z\cap R\neq \emptyset$, we have that
\[ a_n = \sum_{Z_1 \cap X \neq 0} \dots \sum_{\substack{Z_n \cap Z_{n-1} \neq 0 \\ Z_n \cap R \neq 0}} \norm{\mcl M_{Z_1}} \dots \norm{\mcl M_{Z_n}} \]
We have that $a_1$ is bounded by using assumption~\eqref{eq:assumption-a1}:
\begin{multline*}
	a_1 \le \sum_{i \in X} \sum_{\substack{Z_1 \ni i \\ Z_1\cap R \neq 0}} \norm{\mcl M_{Z_1}} \le \\
	\sum_{i \in X} \nu^{-1}(\dist(i,R)) \sum_{\substack{Z_1 \ni i \\ Z_1\cap R \neq 0}} \norm{\mcl M_{Z_1}}  \nu(\diam(Z_1))
	\\ \le v \sum_{i \in X} \nu^{-1}(\dist(i,R)) .
\end{multline*}
Similarly, we bound $a_2$ as follows:
\[
	a_2 \le \sum_{i \in X} \sum_{Z_1 \ni i} \sum_{j \in Z_2} \sum_{\substack{Z_2 \ni j\\ Z\cap Y \neq 0 }} \norm{\mcl M_{Z_1}} \norm{\mcl M_{Z_2}}.
\]
We now use the fact that $\nu(\dist(i,j)) \nu(\dist(j,R)) \ge \nu(\dist(i,R))$, and thus
\[ \begin{split}
	a_2 \le \sum_{i \in X} \nu^{-1}(\dist(i,R)) \sum_{Z_1 \ni i} \norm{\mcl M_{Z_1}} \sum_{j \in Z_2} \nu(\dist(i,j))
		 \sum_{\substack{Z_2 \ni j\\ Z\cap Y \neq 0 }} \nu(\dist(j, Y)) \norm{\mcl M_{Z_2}} \\
	\le \sum_{i \in X} \nu^{-1}(\dist(i,R)) \sum_{Z_1 \ni i} \norm{\mcl M_{Z_1}} \nu(\diam(Z_1)) \sum_{j \in Z_1}
		\sum_{\substack{Z_2 \ni j\\ Z\cap Y \neq 0 }} \nu(\diam(Z_2)) \norm{\mcl M_{Z_2}} \\
	\le v \sum_{i \in X} \nu^{-1}(\dist(i,R)) \sum_{Z_1 \ni i} \norm{\mcl M_{Z_1}} \nu(\diam(Z_1)) \abs{Z_1} \\
	\le v^2 \sum_{i \in X} \nu^{-1}(\dist(i,R)).
\end{split} \]

Proceeding in a similar way, we can bound $a_n$ by $v^n \sum_{i \in X} \nu^{-1}(\dist(i,R))$. Let us now turn our attention to $b_n$.
Let $Z_1 \cap X \neq \emptyset$. Then for all $u$ in $Z_1$ it holds that
\[ \dist(X,R) \le \dist(u,X) + \dist(u,R) \le \diam Z_1 + \dist(u, R) .\]
In particular, this holds for $y_{Z_1} \in Z_1$ such that $\dist(y_{Z_1}, R) = \dist(Z_1, R)$. Therefore we have that
\[ 1 \le \nu(\diam Z_1) \nu (\dist(Z_1,R)) \nu^{-1}(\dist(X,R)) .\]
We can use the previous inequality to bound $b_2$ as follow:
\begin{multline*}
	b_2 = \sum_{Z_1 \cap X \neq \emptyset } \norm{M_{Z_1}} \delta (Z_1) \\
	\le \nu^{-1}(\dist(X,R)) \sum_{x \in X} \sum_{Z_1 \ni x} \norm{M_{Z_1}} \nu(\diam Z_1)
	{\Delta(\dist(Z_1,R))}{\nu(\dist(Z_1,R))}
	\\ \le v \poly(m) \abs{X} \nu^{-1}(\dist(X,R)).
 \end{multline*}
For $b_3$, we reason similarly as follows:
\begin{multline*}
	b_3 = \sum_{Z_1 \cap X \neq \emptyset} \sum_{Z_2 \cap Z_1 \neq \emptyset}  \norm{M_{Z_1}} \norm{M_{Z_2}} \delta (Z_2)
	\\ \le \sum_{Z_1 \cap X \neq \emptyset} \frac{\norm{M_{Z_1}}}{\nu(\dist(Z_1, R))}
		\sum_{z \in Z_1} \sum_{Z_2 \ni z} \norm{M_{Z_2}} \nu(\diam Z_2) {\Delta(\dist(Z_2,R))}{\nu(\dist(Z_2,R))}
	\\ \le v \poly(m) \sum_{Z_1 \cap X \neq \emptyset} \norm{M_{Z_1}} \nu^{-1}(\dist(Z_1, R)) \abs{Z_1}
	\\ \le v \poly(m) \nu^{-1}(\dist(X,R)) \sum_{x \in X} \sum_{Z_1 \ni x} \norm{M_{Z_1}} \abs{Z_1} \nu(\diam Z_1)
	\\ \le v^2 \poly(m) \nu^{-1}(\dist(X,R)) \abs{X}.
\end{multline*}
Following the same argument, we can thus bound the general term $b_{n+1}$ by
\[ v^n \poly(m) \abs{X} \nu^{-1}(\dist(X,R)).\]

We can then bound
\[ \sum_{n} a_n \frac{t^n}{n!} \le (e^{vt}-1) \sum_{i \in X} \nu^{-1}(\dist(i,R)) \le (e^{vt}-1) \abs{X} \nu^{-1}(\dist(X,R))  \]
and
\[ \sum_{n} b_n \frac{t^n}{n!} \le \poly(m) v^{-1} (e^{vt}-1) \abs{X}\nu^{-1}(\dist(X,R)) + t [\delta(X) - \poly(m)\abs{X}\nu^{-1}(\dist(X,R)) ]. \]
Note that, because of Assumption~\eqref{eq:assumption-a2}, the last term in the \rhs can be bounded as
\[ \delta(X) - \poly(m)\abs{X}\nu^{-1}(\dist(X,R)) \le \poly(m)\abs{X}\nu^{-1}(\dist(X,R)) \]
This concludes the proof.
\end{proof}

We can now prove equation~\eqref{eq:localizing}.
\begin{proof}
Applying Duhamel's formula \cite{tao2006nonlinear} we have that
\[ \left( T_t^{ \bar B} - T_t^{ \bar B\setminus A} \right)(\rho_\infty^m \otimes \tau) =
\int_0^t T_{t-s}^{ \bar B \setminus A} \tilde{\mathcal L}^A \rho(s) \de s ;
\]
and therefore:
\[ \norm{ \left( T_t^{\bar B} - T_t^{ \bar B \setminus A} \right)(\rho_\infty^m \otimes \tau) }_1
\le \sum_{Z \cap A \neq \emptyset} \int_0^t \norm{\mathcal M_Z \rho(s)}_1 \de s  +
\sum_{d > m+1} \norm{B_d^{\partial B}}_\diamond t
.\]
The second term on the \rhs is bounded by $t \poly(m) \nu^{-1}(m)$.
Let us focus on the first term on the \rhs
If $Z \subset A(m)$, we can bound the r.h.s with equation~\eqref{eq:lieb-rob}. In particular, we have the following
\[
\sum_{\substack{Z\cap A \neq \emptyset \\ Z\subset A(m)}} \norm{\mathcal M_Z \rho(s) }_1 \le
\poly(m) (e^{vs} +t - 1) \sum_{\substack{Z\cap A \neq \emptyset \\ Z\subset A(m)}} \norm{\mathcal  M_Z}_\diamond \abs{Z} \nu^{-1}(\dist(Z,R)).
\]

Observe that, since $Z \cap A \neq \emptyset$, it holds that $\dist(Z,R) + \diam Z \ge m$, and therefore
\[
\sum_{\substack{Z\cap A \neq \emptyset \\ Z\subset A(m)}} \norm{\mathcal  M_Z}_\diamond \abs{Z} \nu^{-1}(\dist(Z,R)) \le
\nu^{-1}(m) \sum_{\substack{Z\cap A \neq \emptyset \\ Z\subset A(m)}} \norm{\mathcal  M_Z}_\diamond \abs{Z}  \nu(\diam Z) \le \nu^{-1}(m) v \ ,
\]
where we have used the Lieb-Robinson assumption.

If $Z \not \subseteq A(m)$, then it must hold that $Z \cap R \neq \emptyset$. Then we showed in the previous lemma that:
\[ \sum_{\substack{Z\cap A \neq \emptyset \\ Z\cap R \neq \emptyset}} \norm{\mathcal M_Z}_1 \le v \abs{A} \nu^{-1}(\dist(A,R))
 = v \abs{A} \nu^{-1}(m) .\]
Putting it all together, we have that
\[ \norm{ \left( T_t^{\bar B} - T_t^{ \bar B \setminus A} \right)(\rho_\infty^m \otimes \tau) }_1
	\le \poly(m) \nu^{-1}(m) \left[ e^{vt} -1 + t \right] .
 \]
\end{proof}

\section{Decay of correlations}
\label{section:correlation-decay}
In this section we show that, as a straightforward consequence of Lemma~\ref{lemma:rapidmixing-localization}, the hypotheses on $\mathcal L$ imply that its fixed points have a particular character: they have fast decay of correlations, meaning that the correlations between two spatially separated regions is fast-decaying in distance. How fast this decay is is given by the decaying function $\Delta_0$ defined in Lemma~\ref{lemma:rapidmixing-localization}.

\begin{thm}
  Under the same assumption as in Lemma~\ref{lemma:rapidmixing-localization}, fix two regions $A$ and $B \subset \Lambda$, let $d_{AB} > 0$ be the distance between them.
  Then we have that
  \begin{equation}
  \label{eq:correlation-decay}
  T(A : B) \le 3 (\abs{A} + \abs{B})^\delta \Delta_0\left(\frac{d_{AB}}{2}\right) ,
  \end{equation}
  where the correlations are calculated with respect to $\rho_\infty$, and $\delta$ and $\Delta_0$ are defined
  in Lemma~\ref{lemma:rapidmixing-localization}.
\end{thm}

\begin{proof}
  Let $C = A \cup B$,
  and denote by $\rho_{AB}$ the reduced density matrix of $\rho_\infty$ over $C$, and by
  $\rho_A$ and $\rho_B$ the reduced state on $A$ and $B$, respectively.

  Consider $\rho^s_\infty$ the unique fix point of $T_t^{\overline{C(s)}}$, 
  and denote by $\rho_A^s$ and $\rho_B^s$ its reduced density matrices over $A$ and $B$, respectively.
  If $s \le \frac{d_{AB}}{2}$, then $C(s)$ has two disjoint components corresponding to $A(s)$ and $B(s)$,
  and thus $\rho^s_\infty$ decomposes as a tensor product over such bipartition, 
  and its reduced density matrix over $C$ is given by $\rho_A^s \otimes \rho^s_B$.

  By Lemma~\ref{lemma:rapidmixing-localization}, we have that, for any observable $O_C$ with operator norm equal to 1 and supported on $C$,
  \[ \abs{ \trace O_C (\rho_\infty - \rho^{s}_\infty) } \le \abs{C}^\delta \Delta_0(s) .\]
  This implies that
  \[ \norm{\rho_{AB} - \rho_A^s \otimes \rho^s_B }_1 \le \abs{C}^\delta \Delta_0(s).\]
  Since the trace norm does not increases under the partial trace, then
  \[ \norm{\rho_{A} - \rho^s_A }_1 \le \abs{C}^\delta \Delta_0(s) ,\]
  and the same holds for $B$. This in turn implies that
  \[ \norm{\rho_{A}\otimes\rho_{B} -  \rho^s_A \otimes  \rho^s_B}_1 \le 2 \abs{C}^\delta \Delta_0(s) ,\]
  and by applying the triangle inequality we obtain the desired result.
\end{proof}

\begin{obs}
  By equation~\eqref{eq:alicki-mutual}, we have that the mutual information $I(A:B)$ decays with $\dist(A,B)$ at essentially the same rate as $T(A:B)$.
\end{obs}

Note that the dependence on $\abs{A}$ and $\abs{B}$ of the bound in equation~\eqref{eq:correlation-decay}, which is harmless when $A$ and $B$ are finite regions -- as in the case of two-point correlation functions -- becomes significant when one of the two regions is proportional to the system size. In \cite{PhysRevLett.100.070502} the authors defined a correlation length $\xi$ for the mutual information, as the minimal length such that, for all $L \ge \xi$, it holds that
\[ I(b_x(R-L) : b_x(R)^c ) \le \frac 1 2 I( b_x(R) : b_x(R)^c), \quad \forall x, \forall R \ge 0 .\]

From this property they are able to derive an area law for the mutual information of the type:
\[ I(A : A^c ) \le 4 \abs{\partial A} \xi .\]

If we were to use~\eqref{eq:correlation-decay} to determine a correlation length $\xi$, we would obtain a value for $\xi$ which depends on the system size, and thus obtain a bound comparable to equation~\eqref{eq:kastoryano-area-law} obtained in \cite{Kastoryano-exp-decay}. In the next section we want to show that it is possible to greatly improve this bound, at the cost of adding some extra hypotheses on the evolution.


\section{Area law for mutual information}
\label{section:correlation-arealaw}
Until now, we have avoided as much as possible to put restrictions on the function $\nu(r)$ appearing in Assumptions \eqref{eq:assumption-a1} and \eqref{eq:assumption-a2}, and we have simply required it to be fast enough for Lieb-Robinson bounds to hold. Indeed, in most of the bounds obtained in the previous results it appears $\nu(r)$ (or $\Delta_0(r)$, which depends on $\nu(r)$), so that weaker assumptions will simply lead to weaker bounds.

This gets more complicated when it gets to prove the area law bounds, since we are actually interested in pinning down the case in which the bound on mutual information takes the form of $\abs{\partial A} \log \abs{A}$ - a area law with a logarithmic correction. As will be clear in the proof, the logarithmic correction depends on $\nu(r)$ (and consequently $\Delta_0(r)$) to be exponential. Slower rates will still lead to a bound on the mutual information, but where the logarithmic correction is replaced by a super-logarithmic correction. The resulting bound can hardly be called an area law - a sub-volume law would be more correct. We will avoid such generalization, as they make the proof unnecessarily  complicated, and focus on the more interesting case of exponentially decaying interactions, for which we can prove a proper area law with logarithmic correction.

\subsection{The pure fixed point case}
In this section, let us suppose that each $T^{\bar \Lambda}_t$ has a unique fixed point, and moreover that this fixed point is pure. This setting is of particular interest in view of dissipative state engineering \cite{verstraete09,Kraus08}, since ideally one would like to be able to create pure states (at least before noise and errors are taken into account).

Let us denote by $\ket{\phi_\Lambda}$ the pure fixed point of $T^{\bar \Lambda}_t$.

\begin{prop}
Let $\mcl L = \{ \mcl M, \mcl B\}$ be a uniform family of dissipative evolutions,
satisfying rapid mixing and having a unique pure fixed point $\ket{\phi_\Lambda}$ for each $\Lambda$.
Fix $A\subset \Lambda$, and let $\rho_A$ denote the reduced density matrix of $\ket{\phi_\Lambda}$ on $A$.
Then it holds that
\[ S(\rho_A) \le c \log \abs{A} \cdot \abs{\partial A} \]
for some constant $c>0$.
\end{prop}
\begin{proof}
Let $\ell \ge 0$, to be determined later, and denote by $\sigma$ the reduced density matrix of $\ket{\phi_{A(\ell)}}$ on $A$.
Then we have trivially that
\[ S(\sigma) \le \log_2 \dim \mathcal H_{A(\ell) \setminus A} \le c_0\, \ell\, \abs{\partial A}  ,\]
for some positive constant $c_0$.
On the other hand, by lemma~\ref{lemma:rapidmixing-localization}, we have that $d := \norm{\rho_A - \sigma}_1 \le \abs{A}^\delta \Delta_0(\ell)$,
and thus by equation~\eqref{eq:fannes-audenaert}

\[\begin{split}
 S(\rho_A) \le S(\sigma) + \abs{S(\sigma) - S(\rho_A)} \le \\
c_0\, \ell\, \abs{\partial A} + 2 d \abs{A} + 2 h_b(d) \\
\le c_0\, \ell\, \abs{\partial A}  + 2 \abs{A}^{\delta+1} \Delta_0(\ell) + 2 h_b(d)
.
\end{split}\]
Fix a $\epsilon$ such that $0 < \epsilon < 1/2$ and
let us choose $\ell$ such that $ 2 \abs{A}^{\delta+1} \Delta_0(\ell) \le \epsilon$. This implies that $\ell$ scales as $\log \abs{A}$, and thus 
$ S(\rho_A) \le c_1 \log \abs{A}\cdot \abs{\partial A} + \epsilon + 2 h_b(\epsilon/2\abs{A})$,
for some positive $c_1$. By taking $c \ge c_1$ we can absorb the terms depending on $\epsilon$ in the other one, and obtain the claimed estimate $S(\rho_A) \le c \log \abs{A}\cdot \abs{\partial A} $.
\end{proof}

While of interest, the case of a pure fixed point is a very specific one. Therefore, we want to give results applicable in the generic case of a mixed fixed point. To obtain such results we will need to make an additional assumption, namely that the system is frustration free.

\subsection{The frustration-free case}
\begin{thm}[Area law for mutual information]
Let $\mcl L = \{ \mcl M, \mcl B\}$ be a uniform family of dissipative evolutions,
satisfying rapid mixing, frustration freeness and having a unique pure fixed point.
Let $\rho_\infty$ be the fixed point of $\mathcal L^{\bar \Lambda}$ for some $\Lambda = b_u(L) \subset \Gamma$.
Then we have that
\begin{equation}
I(A:A^c)_{\rho_\infty} \le c \Abs{\partial A} \log \abs{A}
\end{equation}
for some positive $c$ independent of the system-size.
\end{thm}
\begin{proof}
For each $n \ge 0$, let $\rho_\infty^n$ be the fixed point of $T_t^{\bar A(n)}$. Fix a positive $n_0$ to be determined later.
Then it holds that
\begin{multline}
\label{eq:area-law-proof}
 I(A:A^c)_{\rho_\infty} = \\
 I(A:A^c)_{\rho_\infty^{n_0}} + \sum_{n = n_0}^{L-1} \left [ I(A:A^c)_{\rho^{n+1}_\infty} - I(A:A^c)_{\rho^n_\infty} \right ].
\end{multline}
We want to show that it is possible to choose $n_0$ in such a way that $I(A:A^c)_{\rho_\infty^{n_0}} \le c \Abs{\partial A} \log \abs{A}$ and
the sum in the \rhs is arbitrarily small.

For each $n \ge 0$, we have that, by rapid mixing~\eqref{eq:rapid-mixing}
\[ \norm{\rho_\infty^{n+1} - T_t^{\bar A(n+1)}(\rho_\infty^n \otimes \tau)}_1 \le \abs{A} \phi_1(n) e^{-\gamma t} ,\]
where $\phi_1(n)$ is a polynomial in $n$.
On the other hand, equation~\eqref{eq:localizing} implies that
\[ \norm{ T_t^{\bar A(n+1)}(\rho_\infty^n \otimes \tau) - T_t^{\bar A(n+1) \setminus A}(\rho_\infty^n \otimes \tau) }_1 \le \abs{A} \phi_2(n)\nu^{-1}(n) [ e^{vt} + t - 1]  ,\]

and $\phi_2(n)$ is polynomial in $n$.
Let us choose choose $t_n$ such that
\[  \epsilon_n := \phi_1(n) e^{-\gamma t_n} + \phi_2(n) \nu^{-1}(n) [e^{vt_n } + t_n -1 ] \]
is exponentially decaying in $n$.
This can be done by taking $t_n$ which scales proportionally to
$\frac{1}{v+\gamma} \log \left( \nu(n) \frac{\phi_1(n)}{\phi_2(n)}\right)$,
which is essentially linear if $\nu(n)$ grows exponentially.
We can put the two bounds together using the triangle inequality, in such a way that
\[ \norm{\rho_\infty^{n+1} - T_{t_n}^{\bar A(n+1) \setminus A}(\rho_\infty^n \otimes \tau)}_1 \le \abs{A} \epsilon_n .\]
Observe that, since $T_t^{\bar A(n+1) \setminus A}$ does not act on $A$,
$ I(A:A^c)_{T_t^{\bar A(n+1) \setminus A}(\rho_\infty^n \otimes \tau)} \le I(A:A^c)_{\rho^n_\infty} .$

Let us assume that $n_0$ is big enough so that $\epsilon_n \abs{A}^2 \le 2^{-n}$ for all $n \ge n_0$,
we can then apply inequality~\eqref{eq:alicki-mutual} and obtain

\begin{multline*}
I(A:A^c)_{\rho_\infty^{n+1}} - I(A:A^c)_{\rho_\infty^n} \le \\ \le
I(A:A^c)_{\rho_\infty^{n+1}} - I(A:A^c)_{T_t^{\bar A(n+1) \setminus A}(\rho_\infty^n \otimes \tau)}  \le \\
 6 \epsilon_n \abs{A}^2 + 4 h_b(\epsilon_n\abs{A}) .
\end{multline*}
Then
\[ I(A:A^c)_{\rho_\infty} \le I(A:A^c)_{\rho_\infty^{n_0}} + 6 \sum_{n = n_0}^{L-1} 2^{-n} + 4 \sum_{n = n_0}^{L-1} h_b(\epsilon_n \abs{A}) .\]

Observe that, if $0\le x \le 1/e$, then $(x-1) \log_2(1-x) \le -x\log_2 x$,
and the latter is a increasing function in that interval. Therefore:
\[ h_b(\epsilon_n \abs{A}) \le - 2 \epsilon_n \abs{A} \log_2 (\epsilon_n \abs{A}) \le
\frac{2^{-n+1}}{\abs{A}} \left(n+\log_2 \abs{A}  \right) .\]
Therefore, $\sum_{n = n_0}^{L-1} h_b(\epsilon_n \abs{A})$ is the tail of a series which is converging geometrically, and therefore is exponentially decaying as $n_0$ increases. The same is true for $\sum_{n = n_0}^{L-1} \lambda^n$, so that both of them can be made smaller than $I(A:A^c)_{\rho_\infty^{n_0}}$.

By taking $n_0$ proportional to $\log \abs{A}$, we can bound $I(A:A^c)_{\rho_\infty^{n_0}}$ by the logarithm of $\dim \mcl H_{A^c}$, which
is proportional to $\abs{A(n_0)\setminus A}$, and therefore:
\[  I(A:A^c)_{\rho_\infty^{n_0}} \le c\, \abs{\partial A}\log \abs{A} .\]
In conclusion, we have bounded the \rhs of~\eqref{eq:area-law-proof} by  $c \abs{\partial A}\log \abs{A}$, and this concludes the proof.
\end{proof}


	\subparagraph*{Acknowledgements}
	\small
	F.\,G.\,S.\,L.\,B.\ is supported by EPSRC.
	T.\,S.\,C.\ is supported by the Royal Society.
	A.\,L.\ is supported by MINECO FPI fellowship BES-2012-052404.
	A.\,L.\ and D.\,P.-G.\ acknowledge support from  MINECO (grant MTM2011-26912), Comunidad de Madrid (grant QUITEMAD+-CM, ref. S2013/ICE-2801) and the European CHIST-ERA project CQC (funded partially by MINECO grant PRI-PIMCHI-2011-1071).
	This work was made possible through the support of grant \#48322 from the John Templeton Foundation. The opinions expressed in this publication are those of the authors and do not necessarily reflect the views of the John Templeton Foundation.
	This project has received funding from the European Research Council (ERC) under the European Union’s Horizon 2020 research and innovation programme (grant agreement No 648913).
	S.\,M.\ acknowledges funding provided by IQIM, an NSF Physics Frontiers Center with support of the Gordon and Betty Moore Foundation through Grant \#GBMF1250, and AFOSR Grant \#FA8750-12-2-0308.
	The authors would like to thank the hospitality of the Isaac Newton Institute for Mathematical Sciences, where part of this work was carried out.
	
	\printbibliography
	
\end{document}